\documentclass[aps,pra,twocolumn,longbibliography,superscriptaddress]{revtex4-1}

\usepackage[T1]{fontenc}
\usepackage{graphicx,amsmath,amssymb,amsfonts,dsfont}
\usepackage{multirow,colortbl}
\usepackage[pdftex]{hyperref}

\usepackage{physics}
\usepackage{bbold}
\usepackage{amsmath, amsthm, amssymb}

\newtheorem{theorem}{Theorem}
\newtheorem{corollary}{Corollary}
\newtheorem{lemma}{Lemma}

\def\ket#1{|#1\rangle}
\def\braket#1#2{\langle#1|#2\rangle}

\def\sectionaps#1{{\it#1}.~---}

\renewcommand{\leq}{\leqslant}
\renewcommand{\geq}{\geqslant}

\makeatletter
\def\CT@@do@color{%
  \global\let\CT@do@color\relax
  \@tempdima\wd\z@
  \advance\@tempdima\@tempdimb
  \advance\@tempdima\@tempdimc
  \advance\@tempdimb\tabcolsep
  \advance\@tempdimc\tabcolsep
  \advance\@tempdima2\tabcolsep
  \kern-\@tempdimb
  \leaders\vrule
  \hskip\@tempdima\@plus  1fill
  \kern-\@tempdimc
\hskip-\wd\z@ \@plus -1fill }
\makeatother

\begin{document}

\title{How much randomness can be generated from a quantum black-box device?}

\author{Marie Ioannou}\affiliation{D\'epartement de Physique Appliqu\'ee, Universit\'e de Gen\`eve, 1211 Gen\`eve, Switzerland}
\author{Jonatan Bohr Brask}\affiliation{D\'epartement de Physique Appliqu\'ee, Universit\'e de Gen\`eve, 1211 Gen\`eve, Switzerland}\affiliation{Department of Physics, Technical University of Denmark, Fysikvej, Kongens Lyngby 2800, Denmark}
\author{Nicolas Brunner}\affiliation{D\'epartement de Physique Appliqu\'ee, Universit\'e de Gen\`eve, 1211 Gen\`eve, Switzerland}

\date{\today}

\begin{abstract}
Quantum theory allows for randomness generation in a device-independent setting, where no detailed description of the experimental device is required. Here we derive a general upper bound on the amount of randomness that can be generated in such a setting. Our bound applies to any black-box scenario, thus covering a wide range of scenarios from partially characterised to completely uncharacterised devices. Specifically, we prove that the number of random bits that can be generated is limited by the number of different input states that enter the measurement device. We show explicitly that our bound is tight in the simplest case. More generally, our work indicates that the prospects of generating a large amount of randomness by using high-dimensional (or even continuous variable) systems will be extremely challenging in practice. 
 \end{abstract}

\maketitle

Randomness is a characteristic feature of quantum theory. The unpredictability of measurements performed on a quantum system have deep implications for information processing, e.g. for quantum random number generation \cite{Herrero-collantes2014,Bera2016,Ma_review}, arguably one of the most developed applications of quantum information science.

The initial idea for devising a quantum random number generator (QRNG) consisted in sending a single quantum particle (say a photon) onto a balanced beam-splitter followed by two detectors \cite{rarity1994,stefanov2000,jennewein2000}. According to quantum theory, it is completely unpredictable on which detector each particle will arrive, thus resulting in a perfectly random bit. The simplicity of this scheme makes it well suited to experimental implementations, and current commercially available QRNGs are mostly based on this principle. In practice, however, the implementation of this scheme is much more challenging than it may appear at first sight. The reason is that any experimental implementation is prone to technical imperfections that introduce unavoidable noise. A rigorous characterization of the devices is therefore required in order to separate technical noise from true quantum randomness, which is often cumbersome and challenging in practice \cite{frauchiger2013,ma2013,haw2015,mitchell2015}. 

Interestingly, however, these problems can in principle be overcome by using a more general approach known as device-independent (DI) certification of quantum randomness. The main feature here is that a detailed description of the experimental devices is not required anymore, and that the user can estimate the amount of randomness generated (i.e. the entropy of the output) based on observed experimental data only, i.e. treating the measurement device as a ``black-box''. Several forms of DI protocols have been considered, featuring different levels of security and practicality. The highest level of security is achieved in the so-called fully DI approach, based on a loophole-free demonstration of quantum nonlocality \cite{colbeckPhD,pironio2010}. Alternative approaches, referred to as semi-DI (SDI), were developed for prepare-and-measure setups, much easier to implement in practice. These schemes typically require a general assumption on the quantum systems involved, for instance an upper bound on the Hilbert space dimension \cite{li2011,li2012} or on the energy \cite{Himbeck}, or a lower bound on the overlap between different states \cite{Brask}. Other approaches to partially DI QNRG have also been investigated, see e.g. \cite{Vallone2014,Law2014,Passaro2015,Cao2015,Cao2016,Marangon2017}.

Currently, there is a strong effort towards the implementation of DI and semi-DI QRNG. State-of-the-art laboratories have demonstrated fully DI QRNGs \cite{pironio2010,christensen2013,NIST,Pan,Kurtsiefer}. Semi-DI QRNG were also realized \cite{lunghi2015,canas2014,Xu2016}, and Ref. \cite{Brask} recently reported performance comparable to commercial devices. An important challenge, which has received considerable attention, is to find schemes that allow for the generation of as much randomness as possible. This then naturally raises the question of what the maximal amount of randomness that can be generated in a black-box scenario is. Here, we address this fundamental question.

Our main result is an upper bound on the amount of randomness that can be generated in any black-box scenario. Notably this bound applies to all scenarios where the measurement device is uncharacterized (i.e. represented by a black-box), hence covering in particular the DI and SDI cases. Specifically, we show that it is not only the number of measurements outcomes that limits the entropy of the output, but also the number of different input states which enter the measurement device. For a measurement device providing $l$ outputs and receiving $k$ different input quantum states, the number of random bits that can be generated is upper bounded by $ \log_2 ( \text{min}\{l,k+1\} )$. Moreover, we show that our bound is tight for the simplest possible scenario. Considering a SDI scheme with only $k=2$ states, and a lower bound on their overlap, we give an explicit scheme where $\log_2(3)$ bits can be certified. Finally, we conclude with a discussion on the implications of our results.

\section{Basics}

For clarity we will present our result in the SDI picture, considering a prepare-and-measure scenario. The preparation device takes as input $x\in\{0, 1, ... , k-1\}$ and emits a quantum state $\rho_{x}$. The emitted state is sent to a measurement device, where a measurement is selected via an input $y\in\{0,...,m-1\}$. The selected measurement is performed and provides an outcome $b\in\{0,...,l-1\}$. The observed statistics is given by the probabilities 
\begin{equation}
p(b|x,y)=\operatorname{Tr}[M_{b|y}\rho_x],
\end{equation}
where $M_{b|y}$ are the element of a positive-operator-valued measure (POVM) describing the quantum measurements. Importantly, in this picture the observer chooses the inputs $x$ and $y$ and records the output $b$, but does not necessarily know what quantum states $\rho_x$ and measurements $M_{b|y}$ are actually being implemented inside the black boxes.

In order to certify randomness in this SDI scenario, one needs to limit the set of possible states $\rho_x$ that can be prepared. If not, all possible distributions $p(b|x,y)$ can be obtained, by simply encoding the input $x$ in a set of $k$ orthogonal quantum states, or equivalently by using $\log_2(k)$ bits of classical communication. Hence, the observed data does not enable any differentiation between classical and quantum behaviours of the devices, and no randomness can be certified. Several possibilities to limit the set of prepared quantum states have been investigated, such as bounds on the Hilbert space dimension, the energy, or the overlap. Here this choice is not important, as our result will apply in full generality, irrespective of which specific assumption is considered. In particular, the states could be completely characterised. We also note that, although we do not explicitly account for classical or quantum side information in the following, such side information can only decrease the amount of certifiable randomness. Hence our upper bound also applies to scenarios with side information.

The next question is how to quantify the amount of randomness that is being generated, that is, how much genuine randomness does the output $b$ contain? This can be done by deriving an upper bound on the probability that any observer (including a potential adversary) has to predict the output $b$. Importantly, the adversary can have complete knowledge of the inner workings of the devices, i.e. know exactly what the prepared states $\rho_x$ and the measurements $M_{b|y}$ are. Typically, works on black-box randomness generation quantify the randomness via the min-entropy $H_{min}=-\log_2(p_{guess})$ \cite{Konig2009}, where $p_{guess}$ is the probability that any observer has to correctly guess the output $b$.

In this work, our goal is to derive a general upper bound on $H_{min}$. Clearly, one must have that $H_{min} \leq \log_2(l)$, as one can always simply guess $b$ at random. It is natural then to ask if this bound can be attained in general. This would be of particular interest for setups where the output alphabet is very large (or even infinite) as in continuous variable (CV) optics implementations, see e.g.~\cite{gabriel2010,Qi2010,symul2011}, thus leading to the certification of a large number of random bits in each round. We will show however that this is not possible in general, as the min-entropy $H_{min}$ depends not only on the properties of the measurement device, but also on the preparation device. Specifically, we show that the number of different preparations $k$ limits the entropy: $H_{min} \leq \log_2(k+1)$.

Before discussing our main result, we introduce some notation. For our analysis, it will be enough to consider finite-dimensional systems, i.e. qudits. These can be conveniently characterized via a generalized Bloch-sphere representation \cite{Aerts2014}:
\begin{equation}
\label{eq.generalisedbloch}
\rho=\frac{1}{d}\left(\mathbb{1}+c_d\vec{n}\cdot\vec{\sigma}\right),
\end{equation}
where $\vec{n}\in\mathbb{R}^{d^2-1}$, $c_d=\sqrt{\frac{d(d-1)}{2}}$ and $\vec{\sigma}$ is a vector of the generalized Gell-Mann matrices (for $d=2$, this a vector of the three Pauli matrices). These $d^2-1$ matrices are traceless and form an orthogonal basis for the space of $d\times d$ Hermitian matrices, i.e. $\operatorname{Tr}[\sigma_i]=0$ and $\operatorname{Tr}[\sigma_i\sigma_j]=2\delta_{i,j}$. From this, it follows that $\rho$ is self-adjoint and has unit trace. However, it is not guaranteed to be positive semidefinite unless further restrictions are placed on $\vec{n}$. For pure states, a simple criterion can be stated in terms of the so-called star product, defined by $(\vec{u}\star\vec{v})_i=\frac{c_d}{d-2}\sum_{j,k=1}^{d^2-1}d_{ijk}u_jv_k$ where $d_{ijk}$ is a symmetric tensor given by the structure constants of the Lie algebra of $SU(d)$. The expression \eqref{eq.generalisedbloch} represents a valid pure state if and only if $|\vec{n}|=1$ and $\vec{n}\star\vec{n}=\vec{n}$ \cite{goyal2016}.

A measurement is represented by a POVM acting on this $d$-dimensional Hilbert space. Considering rank-1 POVMs with $N$ outcomes $\mathbb{P}_N$, there exists a set of positive coefficients $\{\lambda_b\}$ such that $\mathbb{P}_N=\{\lambda_bE_b\}=\{M_b\}$ and $\sum_{b=1}^N\lambda_bE_b=\mathbb{1}$ where $E_b$ are rank-1 projectors (see e.g.~\cite{Sentis2013}). Using the generalized Bloch-sphere representation, $E_b$ can be written as
\begin{equation}\label{POVM}
E_b=\frac{1}{d}\left(\mathbb{1}_d+c_d\vec{v}_b\cdot\vec{\sigma}\right),
\end{equation}
where again $\vec{v}_b\in\mathbb{R}^{d^2-1}$ is a unit vector satisfying $\vec{v}_b\star\vec{v}_b=\vec{v}_b$. The validity of a rank-1 POVM is ensured by
\begin{align} \label{dcond}
\sum_b\lambda_b&=d , \quad \sum_b\lambda_b\vec{v}_b=\vec{0} 
, \quad
\lambda_b\geq0 
\,.
\end{align}

Finally, we will need to consider convex combinations of POVMs, and POVMs with different numbers of outcomes. Given two POVMs $\mathbb{P}^{(1)}$, and $\mathbb{P}^{(2)}$, their convex combination $p\mathbb{P}^{(1)}+(1-p)\mathbb{P}^{(2)}$ is a POVM with the $i^{th}$ element given by $p M_i^{(1)} + (1-p) M_i^{(2)}$, for some $p \in [0,1]$. A POVM is called extremal when it cannot be expressed as a convex combination of other POVMs. Any POVM can be decomposed into extremals, and since convex combinations can be obtained via classical postprocessing, clearly no POVM can generate more randomness than the best extremal entering in its decomposition. Thus, while we consider a scenario with $l$-outcome measurements, it will be interesting to consider POVMs that can be decomposed into extremals with fewer outcomes. By $\mathbb{P}_N$ we denote a POVM with $N$ non-zero elements (and thus $l-N$ zero elements), and by $\mathcal{P}_N$ we denote the set of POVMs which can be written as convex combinations of $N$-outcome POVMs.

\section{Main result}

Our main result is a general upper bound on the amount of randomness that can be generated in a black-box scenario. Below we state and prove the result for the SDI prepare-and-measure scenario. Then we discuss the extension to the fully DI scenario, based on a Bell test.

Let us first give the intuition behind the result. In the black-box scenario, any setup featuring $k$ different prepared quantum states can always be modeled by considering a set of states $\rho_x$ living in a Hilbert space of dimension $d=k$, i.e. $\mathbb{C}^k$. In turn, this implies that the measurement operators $\{M_{b|y}\}$ can also be considered to act on $\mathbb{C}^k$. Now, any POVM acting on $\mathbb{C}^k$ can be simulated from extremal POVMs and classical post-processing. As any extremal POVM acting on $\mathbb{C}^k$ features at most $k^2$ outcomes \cite{Dariano}, it follows directly that no more than $ 2 \log_2(k)$ random bits can be generated per round. 

This simple argument explains why randomness is bounded by the number of possible preparations $k$. However, the specific bound is far from being tight. Intuitively, with only $k$ preparations, their correponding Bloch vectors span a $k$-dimensional real space and any component of the measurements acting outside this space will not contribute to randomness generation, so only a subset of POVMs acting on $\mathbb{C}^k$ will be relevant. Indeed, this is the case, as we show below. For the relevant subset, we find that all extremal POVMs have at most $k+1$ outputs, and it follows that no more than $\log_2(k+1)$ bits can actually be generated.

We note that any POVM element with rank higher than one can alwys be decomposed into a combination of rank-1 operators. By assigning separate outcomes to these operators, one obtains a rank-1 POVM with additional outcomes. The original POVM can be obtained from this larger POVM by classical post-processing (by binning several outcomes together) \cite{Haapsalo2012,Sentis2013}. Since classical post-processing cannot increase the amount of randomness, we can restrict our analysis to rank-1 POVMs.

\begin{theorem}
In a prepare-and-measure setup with $k$ prepared states, and $m$ measurements providing $l$ outputs, one can generate at most $ \log_2 ( \text{min}\{l,k+1\} )$ random bits per round.
\label{mainth}
\end{theorem}
\begin{proof}
The bound $H_{min} \leq \log_2 ( l)$ trivially follows from the fact that an observer can simply guess at random the output $b$. The main aspect of the proof is therefore to show that $H_{min} \leq \log_2 ( k+1)$. Also note that the following arguments holds for any $y$, and thus the bound holds irrespective of $m$.

Given $k$ different prepared quantum states, we can without loss of generality consider that all states $\rho_x$ act on a Hilbert space dimension $d=k$. The statistics can thus be expressed using the generalized Bloch-sphere representation:
\begin{equation}
\begin{aligned}
p(b|x,y)&=\operatorname{Tr}[M_{b|y}\rho_x]\\
&=\operatorname{Tr}\left[\frac{\lambda_{b|y}}{d}\left(\mathbb{1}_d+c_d\vec{v}_{b|y}\cdot\vec{\sigma}\right)\frac{1}{d}\left(\mathbb{1}_d+c_d\vec{n}_x\cdot\vec{\sigma}\right)\right]\\
&=\frac{\lambda_{b|y}}{d}(1+(d-1)\vec{v}_{b|y}\cdot\vec{n}_x).
\end{aligned}
\end{equation}
First, note that the components of $\vec{v}_{b|y}$ orthogonal to $\vec{n}_x$ will not contribute to the statistics. Therefore, it is sufficient in general to consider POVM's whose Bloch-vectors $\vec{v}_{b|y}$ live in the space spanned by $\{\vec{n}_0,...,\vec{n}_{d-1}\}$. Secondly, as $p(b|x,y)$ is linear in $M_{b|y}$, it is sufficient to focus on extremal POVMs. Indeed, if $M_{b|y}$ is not extremal, i.e. it can be written as a convex combination $M_{b|y}=pM_{b|y}^{(1)}+(1-p)M_{b|y}^{(2)}$ with $p\in[0,1]$, then $M_{b|y}$ cannot generate more randomness than $M_{b|y}^{(1)}$ or $M_{b|y}^{(2)}$.

To summarize, we need to focus on extremal rank-1 POVMs with Bloch-vectors $\vec{v}_b$ living in the $d$-dimensional space spanned by $\{\vec{n}_0,...,\vec{n}_{d-1}\}$. Specifically, we would like to determine the maximal number of outputs of any these POVMs. In the following we will show that this maximal number is $d+1$.

In $\mathbb{R}^d$ one needs $d$ vectors to span a solid angle. Given $d+1$ vectors either \textit{(i)} one of them lies in the solid angle spanned by the others and is thus a conical combination of them, or \textit{(ii)} the solid angles spanned by all the possible subsets of $d$ vectors cover the entire $(d-1)$-sphere. Hence, any additional vector will necessarily fall in the solid angle spanned by $d$ of the original vectors and thus be a conical combination of them. Thus, in dimension $d$, given $d+2$ or more vectors, at least one is always a conical combination of $d$ others. The theorem then follows from the following lemma by induction. 
\end{proof}

\begin{lemma}
Given a rank-1 POVM with $l$ outputs $\mathbb{P}_{l}$, if one of the generalized Bloch-vectors is a conical combination of $l'\leq l-1$ of the others, then the POVM can be written as a convex combination of two rank-1 POVMs with $l-1$ outcomes each, i.e. $\mathbb{P}_{l}=p\mathbb{P}_{l-1}^{(1)}+(1-p)\mathbb{P}_{l-1}^{(2)}$.
\label{mainlemma}
\end{lemma}
\begin{proof}
Let us consider a rank-1 POVM with $l$ elements $\mathbb{P}_{l}=\{M_b\}$, $b=0,...,l-1$. The POVM elements are given by $M_b=\lambda_b E_b$ where $E_b$ are expressed in the generalized Bloch-like representation \eqref{POVM}. The parameters $\lambda_b$ and $\vec{v}_b$ satisfy the conditions \eqref{dcond} such that the $M_b$'s form a valid POVM.

First, the operation consists in extracting $\mathbb{P}_{l-1}^{(1)}$ from $\mathbb{P}_{l}$. Without loss of generality, we make the assumption that $\vec{v}_0$ is a conical combination of $l-1$ vectors, 
\begin{equation}
\vec{v}_0=\sum_{b=1}^{l-1}c_b \vec{v}_b
\end{equation}
with $0\leq c_b$. The parameters $\lambda_b^{(1)}$ and $\vec{v}_b^{(1)}$ of $M_b^{(1)}$ are given by
\begin{equation}
\begin{aligned}
\vec{v}_b^{(1)}&=\vec{v}_b , \quad
\lambda_0^{(1)}=0 , \quad
\lambda_b^{(1)}=\frac{1}{N}(\lambda_b+\lambda_0c_b),
\end{aligned}
\end{equation}
where $N$ is a normalization coefficient to be fixed in order to satisfy the first condition in \eqref{dcond}. The second condition in \eqref{dcond} is also fulfilled,
\begin{equation}
\begin{aligned}
\sum_{b=0}^{l-1}\lambda^{(1)}_b\vec{v}_b&=\lambda_0^{(1)}\vec{v}_0+\sum_{b=1}^{l-1}\frac{1}{N}(\lambda_b+\lambda_0c_b)\vec{v}_b\\
&=\frac{1}{N}\left(\sum_{b=1}^{l-1}\lambda_b\vec{v}_b+\lambda_0\sum_{b=1}^{l-1}c_b\vec{v}_b\right)\\
&=\frac{1}{N}\left(-\lambda_0\vec{v}_0+\lambda_0\vec{v}_0\right)=\vec{0}.
\label{zerocond1}
\end{aligned}
\end{equation}
The last condition is straightforward to verify, i.e. $\lambda_b^{(1)}\geq0$. The first step is done, $\mathbb{P}_{l-1}^{(1)}=\{M_b^{(1)}\}$ is a valid POVM.

Next, the coefficient of the convex combination $p$ is defined as follows
\begin{equation}
p=\min_b \frac{\lambda_b}{\lambda_b^{(1)}}.
\label{convexcoeff}
\end{equation}
Here, $p\in[0,1]$ since $\sum_{b=0}^{l-1}\lambda_b=\sum_{b=0}^{l-1}\lambda_b^{(1)}=d$.

Finally, $\mathbb{P}_{l-1}^{(2)}$ can be fixed by defining the parameters of $M_b^{(2)}$ as 
\begin{equation}
\begin{aligned}
\vec{v}_b^{(2)}&=\vec{v}_b  , \quad
\lambda_b^{(2)}=\frac{\lambda_b-p\lambda_b^{(1)}}{1-p} .
\end{aligned}
\end{equation}
Assuming that the minimum of \eqref{convexcoeff} occurs for $b^\ast$, this implies $\lambda_b^\ast=0$ and thus a POVM with $m-1$ outcomes. Let us check that the first condition of \eqref{dcond} is fulfilled 
\begin{equation}
\sum_{b=0}^{l-1}\lambda_b^{(2)}=\frac{1}{1-p}(d-pd)=d.
\end{equation}
Using \eqref{dcond} and \eqref{zerocond1} it is straightforward to verify that
\begin{equation}
\sum_{b=0}^{l-1}\lambda_b^{(2)}\vec{v}_b^{(2)}=\sum_{b=0}^{l-1}\frac{\lambda_b-p\lambda_b^{(1)}}{1-p}\vec{v}_b=\vec{0}.
\end{equation}
The positivity of $\lambda_b^{(2)}$ is ensured by the choice of $p$ \eqref{convexcoeff}. Hence, $\mathbb{P}_{l-1}^{(2)}=\{M_b^{(2)}\}$ is also a valid POVM. And by construction, $\mathbb{P}_{l}$ is a convex combination of the two extracted POVM's, $\mathbb{P}_{l}=p\mathbb{P}_{l-1}^{(1)}+(1-p)\mathbb{P}_{l-1}^{(2)}$.
\end{proof}

Theorem 1 is also relevant in the context of randomness generation in the fully DI scenario. Consider a Bell test with two spatially separated parties, Alice and Bob, sharing a quantum state $\rho\in \mathcal{C}^{d}\otimes\mathcal{C}^{d}$. Upon receiving an input $x$ for Alice and $y$ for Bob, they output $a$ and $b$ respectively. When Alice performs measurement $x$ and obtains output $a$ $M_{a|x}$, Bob's system is steered into the (unnormalised) state
\begin{equation}
\sigma_{a|x}=\operatorname{Tr}_A[\rho (M_{a|x}\otimes\mathbb{1}_B)],
\end{equation}
where $M_{a|x}$ is the POVM element for Alice's measurement. Bob can thus receive at most $\vert x \vert{\cdot}\vert a \vert$ different states. From Theorem 1, it then directly follows that:
\begin{corollary}
Consider a Bell scenario with Alice having $|x|$ inputs and $|a|$ outputs, and Bob any number of inputs with $|b|$ outputs. Then, Bob's measurement can generate at most $\log_2( \text{min}\{ |b|, \vert x \vert{\cdot}\vert a \vert+1 \}  )$ random bits per round.
\label{maincorollary}
\end{corollary}

\section{Tight bound for two preparations}

Our main result, Theorem 1, is a general upper bound on the output entropy that can be certified. It is thus natural to ask whether this bound is tight. Here we consider the simple case of a SDI prepare-and-measure setup with $k=2$ preparations and a single ternary measurement, i.e. with outputs $b\in\{0,1,2\}$. We show that the maximal number of certifiable random bits $H_{min}=\log_2(3)$ can be achieved.

\begin{figure}[b]
\begin{center}
\includegraphics[width=\columnwidth]{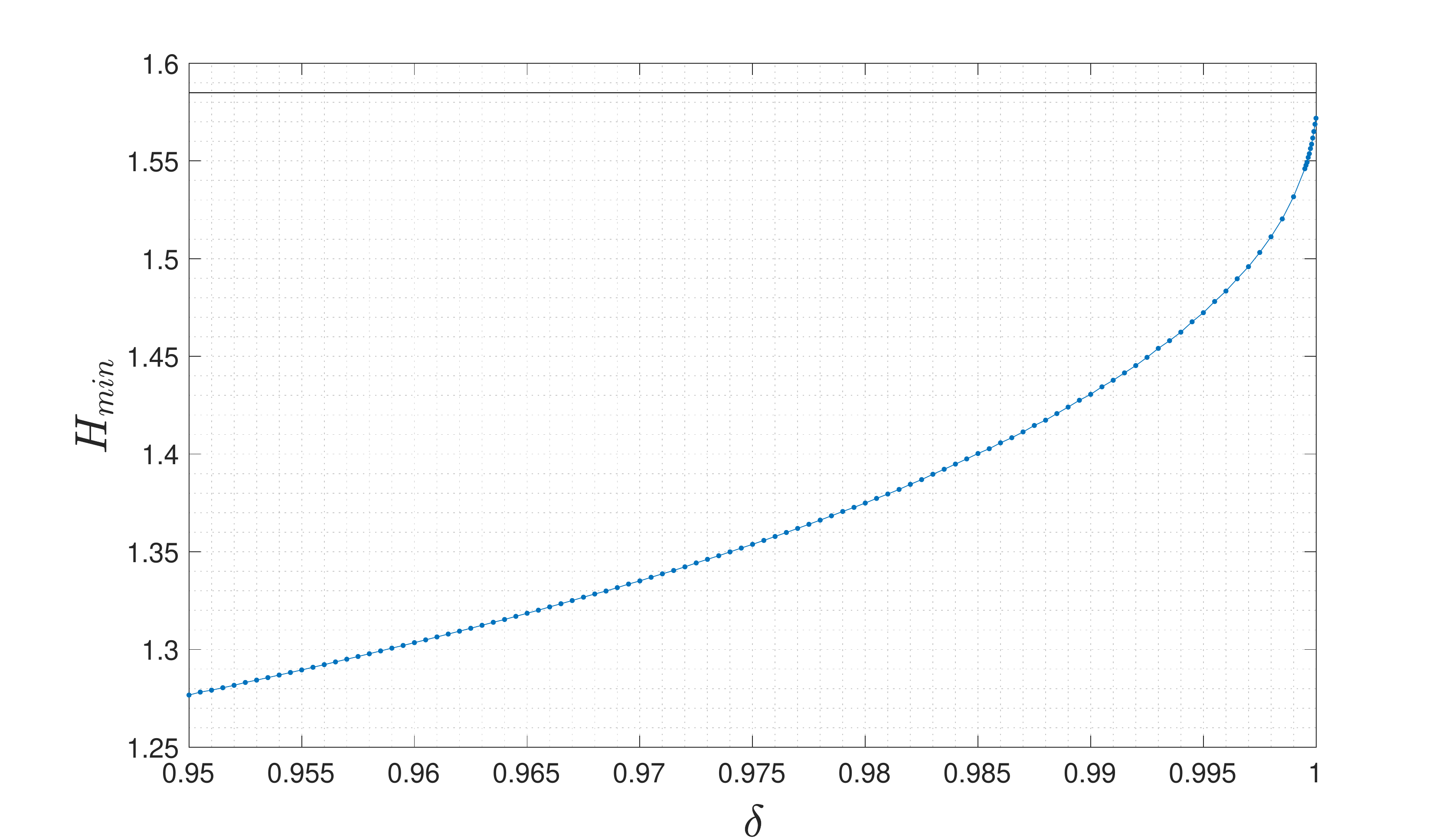}
\caption{Plot of the output entropy $H_{min}$ as a function of the overlap $\delta$ of the two prepared states. In the limit of almost indistinguishable states $\delta \rightarrow 1$, the output entropy becomes maximal, i.e. $H_{min}=\log_2(3)$ (horizontal line). This shows that the bound of Theorem 1 can be attained in this case.}
\label{optimalstrategy}
\end{center}
\end{figure}

Consider a preparation device emitting two possible states $\ket{\psi_x}$ with $x\in\{0, 1\}$. Following Ref. \cite{Brask} we consider an assumption on the distinguishability of the two states, specifically we lower bound their overlap $|\braket{\psi_0}{\psi_1}| \geq \delta$. Note that such an assumption is well suited for optical setups, as it corresponds to an upper bound on the intensity of the light source.

Without loss of generality, the two qubit preparations are given by Bloch vectors $\{\vec{n}_0,\vec{n}_1\}$ in the $xz-$plane of the Bloch sphere, distributed symmetrically around the $z$ axis. From Theorem \ref{mainth}, we can focus on extremal ternary POVMs $\mathbb{P}_3=\{M_1,M_2,M_3\}$ such that all Bloch vectors are in the the $xz-$plane. Specifically, we consider POVMs of the form \begin{align}
M_b&=\frac{\lambda_b}{2}(\mathbb{1}_2+\vec{\sigma}\cdot\vec{u}_b),
\end{align}
such that $\sum_{b=0}^2\lambda_b=2$, $\label{sum2}$, $\sum_{b=0}^2\lambda_b\vec{u}_b=0$ and $\lambda_b\geq0$. Moreover, all Bloch vectors are of the form $\vec{u}_b=(\cos\theta_b,0,\sin\theta_b)$ with $|\vec{u}_b|=1$ (as the POVM is extremal).

Next, a maximisation of the entropy $H_{min}$ is performed over the free parameters, namely $\theta_1$, $\theta_2$ and $\lambda_1$, for different values of the overlap $\delta$. This is implemented as a heuristic optimization over the free parameters; for each set of parameters, a lower bound on the entropy is obtained via a semi-definite program, as in Ref. \cite{Brask}. We find that, in the regime where the two states become almost indistinguishable (i.e. $\delta \rightarrow 1$), the entropy approaches $H_{min}=\log_2(3)$. This shows that the bound of Theorem 1 is tight in this case. The optimal POVM can be parametrized as follows: $\theta_1=0$, $\lambda_2=\lambda_3 = \lambda$ $\theta_2=-\theta_3=\arccos(\frac{1}{\lambda}-1)$, $\lambda_1=2(1-\lambda)$, where 
\begin{equation}
\lambda = \frac{0.7323\delta^3-6.077\delta^2+4.017\delta+5.742}{\delta^3-7.645\delta^2+4.903\delta+7.147}.
\end{equation}

Finally, as this setup can achieve $H_{min}>1$ while using only two preparations, it allows one to perform randomness expansion, i.e. the amount of output randomness is larger than the one used for generating the input $x$.

\section{Discussion}

We presented an upper bound on the amount of randomness that can be generated in a black-box scenario. This bound is given by the number of different quantum states that enter the measurement device, irrespective of whether these states are fully, partially, or uncharacterised, and holds with and without classical or quantum side information. Hence, even when considering measurements with a large number of outputs (or even infinite as in CV systems), the amount of randomness that be generated is still limited by the source, specifically by the number of different preparations. The number of preparations required scales exponentially with the number of random bits to be certified per round. Thus, while generating a large number of random bits per round is in theory possible, this would be challenging in practice. Indeed, in any experiment, the number of rounds is finite which in turn limits the number of possible different preparations (even more so if good statistics is required). For instance, in order to generate 10 random bits per round (not even including here randomness extraction), more than $10^3$ different preparations would be required.

\sectionaps{Acknowledgments}
We thank Matt Pusey for helpful comments. Financial support by the Swiss National Science Foundation (Starting grant DIAQ, Bridge project ``Self-testing QRNG'', NCCR-QSIT) and the EU Quantum Flagship project QRANGE is gratefully acknowledged.

\bibliography{USDrandomness}

\end{document}